\documentclass[11pt]{article}
\usepackage[a4paper,margin=1in,footskip=0.25in]{geometry}
\usepackage{graphicx}
\usepackage{epsf}
\usepackage{epsfig}
\usepackage{epstopdf}
\usepackage{amsfonts}
\usepackage{amssymb}
\usepackage{amsmath}
\usepackage{latexsym}
\usepackage[all]{xy}
\usepackage{fix-cm}
\usepackage{placeins}
\usepackage{hyperref}
\usepackage{nameref}
\usepackage{textcomp}

\newtheorem{theorem}{Theorem}
\newtheorem{fact}{Fact}
\newtheorem{definition}[theorem]{Definition}
\newtheorem{proposition}[theorem]{Proposition}
\newtheorem{lemma}[theorem]{Lemma}
\newtheorem{claim}[theorem]{Claim}
\newenvironment{proof}{\paragraph{Proof:}}{\hfill$\square$}

\newcommand{\dbl}{[\![}
\newcommand{\dbr}{]\!]}

\newcommand{\dotcup}{\mathbin{\dot\cup}}

\newcommand{\Log}{\mbox{{\sf L}}}

\newcommand{\ShP}{\mbox{{\sf \#P}}}

\newcommand{\Pt}{\mbox{{\sf P}}}
\newcommand{\NP}{\mbox{{\sf NP}}}

\newcommand{\Gl}{\mbox{{\sf GapL}}}

\newcommand{\NC}{\mbox{{\sf NC}}}

\newcommand{\TCz}{\mbox{{\sf TC}$^0$}}

\newcommand{\NCt}{\mbox{{\sf NC}$^2$}}

\newcommand{\ShSACo}{\mbox{{\sf \#SAC}$^1$}}
\newcommand{\ShSACz}{\mbox{{\sf \#SAC}$^0$}}
\newcommand{\GSACo}{\mbox{{\sf GapSAC}$^1$}}

\newcommand{\poly}{\mbox{{\sf poly}}}
\newcommand{\VP}{\mbox{{\sf VP}}}
\newcommand{\VNP}{\mbox{{\sf VNP}}}

\newcommand{\calC}{\mbox{${\cal C}$}}

\newcommand{\mso}{\mbox{${\sf MSO}$\ }}

\newcommand{\val}[1]{\mbox{val}(#1)}

\newcommand{\comments}[2]{}

\title{Counting Euler Tours in Undirected Bounded Treewidth Graphs}
\usepackage{authblk}

\author[1,3]{Nikhil Balaji}
\author[1]{Samir Datta}
\author[2]{Venkatesh Ganesan}
\affil[1]{Chennai Mathematical Institute, Chennai, India. \texttt{<nikhil,sdatta>@cmi.ac.in}}
\affil[2]{Birla Institute of Technology and Science - Pilani, India \texttt{venkatesh920@gmail.com}}
\affil[3]{Indian Institute of Technology, Bombay, India. \texttt{nbalaji@cse.iitb.ac.in}}

\begin{document}
\maketitle
\begin{abstract}
We show that counting Euler tours in undirected bounded tree-width graphs is
tractable even in parallel - by proving a $\ShSACo \subseteq \NCt \subseteq \Pt$
upper bound. This is in stark contrast to \ShP-completeness of the same problem
in general graphs. 

Our main technical contribution is to show how (an instance of) dynamic 
programming on bounded \emph{clique-width} graphs can be performed efficiently 
in parallel. Thus we show that the sequential result of Espelage, Gurski and 
Wanke \cite{EGW01} for efficiently computing Hamiltonian paths in bounded 
clique-width graphs can be adapted in the parallel setting to count the number 
of Hamiltonian paths which in turn is a tool for counting the number of Euler 
tours in bounded tree-width graphs.  Our technique also yields parallel 
algorithms for counting longest paths and bipartite perfect matchings in 
bounded-clique width graphs.

While  establishing that counting Euler tours in bounded tree-width graphs
can be computed by non-uniform monotone arithmetic circuits of polynomial 
degree (which characterize \ShSACo) is relatively easy, establishing a uniform
$\ShSACo\ $ bound needs a careful use of polynomial interpolation. 
% We further adapt a technique due to Ben-Or and Cleve~\cite{BOC} in 
% circuit complexity to bring down the complexity to \Gl.
\end{abstract}
% \newpage
\section{Introduction}
An Euler tour of a graph is a closed walk on the graph that traverses every
edge in the graph exactly once. Given a graph, deciding if there is an
Euler tour of the graph is quite simple. Indeed, the famous K\"onigsberg
bridge problem that founded graph theory is a question about the existence of an
Euler tour using each of these bridges exactly once. Euler settled this question in the negative
and in the process gave a necessary and sufficient condition for a graph to
be \textit{Eulerian} (A connected graph is Eulerian if and only if all the vertices
are of even degree). This gives a simple algorithm to check if a graph is
Eulerian.

An equally natural question is to ask for the number of distinct Euler
tours in a graph. For the case of directed graphs, the BEST theorem due
to De Bruijn, Ehrenfest, Smith and Tutte gives an exact formula that gives
the number of Euler tours in a directed graph~\cite{BE87,ST41}
which yields a polynomial time algorithm via a determinant computation.
For undirected graphs, no such closed form expression is known and the
computational problem is \ShP-complete~\cite{BW}. In fact, the problem
is \ShP-complete even when restricted to $4$-regular planar graphs~\cite{GS}. 
So exactly computing the number of Euler tours is not in
polynomial time unless $\ShP = \Pt$.

In this paper, we are concerned with the problem of counting  Euler tours on
graphs of bounded treewidth. Many problems which are $\NP$-hard for general
graphs, can be solved in polynomial time on bounded treewidth graphs. Indeed, a
result of Courcelle~\cite{courcelle} asserts that any graph property that is expressible
in Monadic Second Order logic (with edge quantifiers) can be solved in linear time on
bounded treewidth graphs. Elberfeld et al.~\cite{EJT} adapt the theorem of Courcelle in 
the parallel setting and prove a \Log\ bound. However, Eulerianity is provably
not \mso-expressible~\cite{EF95} and hence the approaches mentioned above are not directly
applicable in our context.

Our strategy to count Euler tours is as follows: Given a bounded treewidth graph $G$,
we count the number of Euler tours of $G$ by counting the number of Hamiltonian tours of the
line graph of $G$, $L(G)$. In general, there is no bijection between these two quantities, but
we show that $G$ can be modified to obtain $G'$ ($tw(G')  \leq tw(G) + 3$) such that $G$, $G'$ have 
the same number of Eulerian tours, which equals the number of Hamiltonian tours of $L(G')$.
Henceforth, we will be primarily interested in line graphs of bounded treewidth
graphs. It is known that such graphs are of bounded clique-width~\cite{SO}

We base our proof on a proof that the decision version of Hamiltonicity is polynomial 
time computable in bounded clique-width graphs~\cite{EGW01}.
We prove that this algorithm can be parallelised and extended to the counting version.
Next, we show that for line graphs of bounded tree-width graphs which form the family of interest,
the clique-width expression can be inferred from the corresponding tree
decomposition. The tree decomposition itself is obtainable by the \Log-version of 
Bodlaender's theorem~\cite{EJT}. 

Our main tool in establishing a uniform \NC-bound for counting Hamiltonian cycles on 
bounded clique-width graphs hinges on \emph{polynomial interpolation}.
While polynomial interpolation has been used successfully 
to compute various graph polynomials~\cite{MRAG}, our use is somewhat indirect and subtle:
it is used by the uniformity machine to populate a table whose entries do not 
depend on the input bounded clique-width expression but only the number of vertices
in the corresponding graph and the clique-width. We then build a monotone arithmetic
circuit that uses the clique-width expression of the graph and entries from this table 
to count the number of Hamiltonian cycles in the clique-width bounded graph. We then 
observe that since the number of distinct Hamiltonian tours of a graph is at most 
exponential in the number of vertices of the graph, and the circuit is monotone,
the formal degree of the circuit must be a polynomial in the size of the input graph.
This allows us to use a result from circuit complexity~\cite{AJMV98} to yield an upper 
bound of \ShSACo\ on the complexity of counting Euler tours on bounded treewidth graphs.

Our techniques also yield a parallel upper bound on the problems of counting
longest paths/cycles and counting bipartite matchings in bounded clique-width graphs.
These are well known problems (and \ShP-complete in general graphs) but their 
(counting) complexity has not been investigated in bounded clique-width graphs. 
While~\cite{DDN} studies the problem of counting longest paths and perfect matchings in bounded 
\emph{tree-width DAGs}, we improve the results by resolving
the problems for bounded \emph{clique-width graphs}
% \footnote{Note that
% our result gives a $\Gl\ $ bound for these problems for bounded clique-width graphs that 
% are line graphs of bounded treewidth graphs. For general bounded clique-width graphs, we have 
% an incomparable (with $\Gl\ $) bound of $\ShSACo$, given the clique decomposition.} 
at the cost of replacing the \Log\ bound by a \ShSACo\ bound where we know that 
$\Log\ \subseteq  \ShSACo\ \subseteq \NCt\ \subseteq \Pt$ \footnote{Note that $\ShSACo\ $ is a 
function class and when we say $\ShSACo\ \subseteq \NCt$, what we actually mean is that any bit of
the $\ShSACo\ $ function family of interest is computable by a $\NCt\ $ circuit family}.

\subsection{Previous Work}
Chebolu, Cryan, Martin  have given a polynomial time algorithm for
counting Euler tours in undirected series-parallel graphs~\cite{CCMsp} and 
they have claimed to extend it to a polynomial time algorithm~\cite{CCMtw} for the
counting Euler tours  in bounded tree-width graphs. We would like to point out that
the only incomplete, unrefereed manuscript available publicly~\cite{CCMtw}
sketches an algorithm that does dynamic programming directly on the tree-decomposition.
Since we show how to obtain the line graph of the bounded tree-width graph efficiently
in parallel and then work on this bounded clique-width graph - our approach
is fundamentally different from that of~\cite{CCMsp,CCMtw}. Another difference
is that their algorithm is not designed to be parallelisable.

Also notice that in a precursor to this paper~\cite{BalajiDatta},
using totally different techniques (basically
applications of the Logspace version of Courcelle's theorem~\cite{EJT})
it was claimed that the number of Euler tours in bounded tree-width
directed and undirected graphs can be counted in Logspace but the approach had 
a serious flaw in the undirected version. Later versions~\cite{BalajiDatta2,BD}
of the paper claim the result only for directed graphs. This work proves a 
slightly weaker version of the result - the upper bound being \ShSACo\ rather
than Logspace.

Given that counting Hamiltonian cycles on bounded clique-width graphs will
suffice for our purposes, one result that is directly relevant is that of
Flarup and Lyaudet~\cite{FL}: They study the expressive power of Perfect
Matching and Hamiltonian polynomials of graphs of bounded clique-width
and show that they can simulate arithmetic polynomials, and are
themselves contained in $\VP$. This yields a $\GSACo$ bound (implicit) for counting
Hamiltonian cycles in bounded clique-width graphs right away. There are two
aspects in which the work of~\cite{FL} differs from our work: Firstly, 
even though their techniques are also inspired from~\cite{EGW01} like ours,
they work with a slightly different notion of clique-width namely 
$W-m-\mbox{clique-width}$\footnote{These are weighted versions of clique-width
and are used to produce weighted graphs. ~\cite{FL} motivate this variant of 
clique-width by observing that since $K_n$ has clique-width $2$, most graph
polynomials are $\VNP$-complete for bounded clique-width graphs.}. Secondly,
in the case of counting Euler tours, from a straight-forward application of~\cite{FL}
the best upper bound that can be obtained from the circuit
families constructed in~\cite{FL} is non-uniform $\GSACo$, whereas we get an upper bound of 
Logspace-uniform\footnote{In an earlier version of this paper, we had erroneously
claimed a \Gl\ upper bound for counting Euler tours. As pointed out to us by Ramprasad Saptharishi,
there is a rather serious gap with this approach.} $\ShSACo$.

%% comparison with makowsky's paper 
There is some similarity that this work bears with that of Makowsky et al.~\cite{MRAG}, in that
both involve polynomial interpolation to count witnesses for a graph theory problem. The
similarity is somewhat superficial because we use interpolation to obtain numbers 
\emph{independent} of the input graph while they interpolate to compute a graph polynomial
that crucially depends on the graph. The choice of graph theory problems is also quite different.
In particular,~\cite{MRAG} does not address the Hamiltonian cycle problem.
\subsection{Our Results}
This is the main theorem of this work:
\begin{theorem}\label{thm:main} \#Hamiltonian Cycles (or Paths) for bounded clique-width graphs
is in \ShSACo. Consequently, \#Euler Tours for bounded tree-width graphs is also in \ShSACo.
\end{theorem}
As a bonus we also get the following :
\begin{theorem} \label{thm:bonus}
The following counts can be obtained in \ShSACo\ for bounded clique-width graphs (given a bounded
clique-width expression for the graph):
\begin{enumerate}
\item \#Hamiltonian Cycles
\item \#Longest Paths/Cycles
\item \#Cycle Covers
\item \#Perfect Matchings (for bipartite graphs)
\end{enumerate}
\end{theorem}

\subsection{Overview of Algorithm}
Every Euler tour in a graph yields a Hamiltonian cycle in its line graph. 
Though this map is not bijective we show that we can make it so by altering
the input graph slightly. It is well known~\cite{GW07} that the line graphs
of bounded tree-width graphs have bounded clique-width. We show how to
obtain a bounded clique-width decomposition for the line graph of a bounded
tree-width graph in Logspace using the Logspace version of Courcelle's
Theorem~\cite{EJT} by first obtaining a bounded tree-width decomposition via
a Logspace version of Bodlaender's theorem~\cite{EJT}.

Our main algorithm replaces the sequential procedure from~\cite{EGW01} to decide
if a bounded clique-width graph has a Hamiltonian path. Instead, it computes 
the number of Hamiltonian cycles.  The procedure uses elementary counting coupled
with polynomial interpolation to compute some matrices which are independent
of the input graph depending only on its size.
The matrices are then combined with vectors maintaining counts, along the
structure tree of the clique-decomposition. A degree bound for the monotone
arithmetic circuit then suffices to prove the \ShSACo\ bound.

\subsection{Organization}

The rest of the paper is organized as follows: In Section~\ref{sec:prelims},
we introduce some definitions and results that will be helpful in understanding
the rest of the paper. Section~\ref{sec:preproc} shows how to obtain a clique-width
expression for the line graph of a bounded treewidth graph in Logspace. Section~\ref{sec:algo}
presents a $\ShSACo\ $ implementation of our algorithm to count the number of 
Hamiltonian tours in graphs of bounded clique-width. 
We conclude with some unresolved questions related to this work in Section~\ref{sec:concl}.

\section{Preliminaries}\label{sec:prelims}

\begin{definition}[Line Graph]
 For an undirected graph $G = (V, E)$, the line graph of $G$ denoted $L(G) = (L_V, L_E)$
 is the graph where $L_V = E$ and $(e_i, e_j) \in L_E$ if and only if there exists a
 vertex $v \in V$ such that both $e_i$ and $e_j$ are incident on $v$.
\end{definition}

 \begin{definition}[Treewidth]\label{def:tdecomp}
 Given an undirected graph $G = (V_G, E_G)$ a tree decomposition of $G$ is a tree
 $T = (V_T, E_T)$(the vertices in $V_T \subseteq 2^{V_G}$ are called \textit{bags}), 
 such that 
 \begin{enumerate}
  \item Every vertex $v \in V_G$ is present in at least one bag, i.e., 
  $\cup_{X \in V_T} X = V_G$.
  \item If $v \in V_G$ is present in bags $X_i, X_j \in V_T$, then
  $v$ is present in every bag $X_k$ in the unique path between $X_i$
  and $X_j$ in the tree $T$.
  \item For every edge $(u, v) \in E_G$, there is a bag $X_r \in V_T$ such that
  $u, v \in X_r$.
 \end{enumerate}
 The width of a tree decomposition is $\max_{X \in V_T} (|X| - 1)$. The treewidth of
 a graph is the minimum width over all possible tree decomposition of the graph.
 \end{definition}

\begin{definition}[NLC-width]
Let $k$ be a positive integer. The class $NLC_k$ of labeled graphs 
$G = (V, E, lab_G)$ where $lab_G : V \to [k]$, is recursively defined as follows:
 \begin{enumerate}
  \item The single vertex graph labeled by a label $a$, $\bullet_a$ for $a \in [k]$ is in $NLC_k$.
  \item Let $G = (V_G, E_G, lab_G) \in NLC_k$ and $H = (V_H, E_H, lab_H) \in NLC_k$ be 
  two vertex-disjoint labeled graphs and $S \subseteq [k]^2$, then 
  $G \times_S H = (V', E', lab') \in NLC_k$, where $V' = V_G \cup V_H$ and
  \[
   E' = E_G \cup E_H \cup \{(u,v) | u \in V_G, v \in V_H, (lab_G(u), lab_H(v)) \in S\}
  \]
and for all $u \in V'$,
  \begin{equation*}
   lab'(u) =
   \begin{cases}
      lab_G(u), & \text{if}\ u \in V_G \\
      lab_H(u), & \text{if}\ u \in V_H
    \end{cases}
  \end{equation*}
  \item  Let $G = (V_G, E_G, lab) \in NLC_k$ and $R: [k] \to [k]$ be a function, then
  $\circ_R(G):= (V_G, E_G, lab')$ defined by $lab'(u) = R(lab(u))$ for all $u \in V_G$
  is in $NLC_k$.
  \end{enumerate}
The NLC-width\footnote{NLC stands for \emph{Node Label Controlled}, has its origins in graph
grammars, was defined by Wanke~\cite{Wan94}} of a labeled graph $G$ is the least integer $k$ such that
$G \in NLC_k$. An expression $Y$ built with $\bullet_a, \times_S, \circ_R$, 
for integers $a \in [k]$, $S \in [k]^2$ and $R:[k] \to [k]$ is called a 
NLC-width $k$ expression. The graph defined by expression $Y$ is denoted
by $val(Y)$.
\end{definition}

\begin{definition}[Clique Width]\label{def:cw}
 Let $k$ be a positive integer. The class $CW_k$ of labeled graphs 
 $G = (V, E, lab_G)$ where $lab_G : V \to [k]$ is recursively defined as follows:
 \begin{enumerate}
  \item The single vertex graph labeled by a label $a$, $\bullet_a$ for $a \in [k]$ is in $CW_k$.
  \item Let $G = (V_G, E_G, lab_G) \in CW_k$ and $H = (V_H, E_H, lab_H) \in CW_k$ be 
  two vertex-disjoint labeled graphs. Then $G \oplus H = (V', E', lab') \in CW_k$, where 
  $V' = V_G \cup V_H$ and $E' = E_G \cup E_H$ and for all $u \in V'$
  \begin{equation*}
   lab'(u) =
   \begin{cases}
      lab_G(u), & \text{if}\ u \in V_G \\
      lab_H(u), & \text{if}\ u \in V_H
    \end{cases}
  \end{equation*}
 \item Let $a, b$ be distinct positive integers and $G = (V_G, E_G, lab) \in CW_k$
  be a labeled graph. Then, 
  \begin{enumerate}
   \item[(a)] $\rho_{a \to b}(G) := (V_G, E_G, lab') \in CW_k$ where for all $u \in V_G$
  \begin{equation*}
   lab'(u) =
   \begin{cases}
      lab_G(u), & \text{if}\ lab_G(u) \neq a \\
      b, & \text{if}\ lab_G(u) = a
    \end{cases}
  \end{equation*}
   \item[(b)] $\eta_{a, b}(G) := (V_G, E', lab_G) \in CW_k$ where, 
  \begin{equation*}
   E' = E_G \cup \{(u,v) | u,v \in V_G, lab(u) = a, lab(v) = b\}
  \end{equation*}
  \end{enumerate}
\end{enumerate}
The clique-width of a labeled graph $G$ is the least integer $k$ such that
$G \in CW_k$. An expression $X$ built with $\bullet_a, \oplus, \rho_{a \to b}, 
\eta_{a,b}$ for integers $a, b \in [k]$ is called a clique-width $k$ expression.
%The graph defined by expression $X$ is denoted by $val(X)$. 
By $val(X)$, we denote the graph defined by expression $X$.
\end{definition}

\begin{definition}[Chordal graph, Chordal completion]\label{def:chordal}
A graph is said to be chordal if every cycle with at least $4$ vertices always
contains a chord.  A chordal completion of a graph $G$ is a chordal graph with the same
vertex set as $G$ which contains all edges of $G$.
\end{definition}

\begin{definition}[Perfect Elimination Ordering, Elimination Tree~\cite{golumbic}]
Let $G = (V, E)$ be a graph and $o = (v_1, v_2, \ldots, v_n)$ be an ordering
of the vertices of $G$. Let $N^-(G,o,i)$ and $N^+(G,o,i)$ for $i=1, \ldots, n$
be the set of neighbors $v_j$ of vertex $v_i$ with $j < i$ and $j > i$ respectively.
\begin{eqnarray*}
N^-(G,o,i) &=& \{v_j |(v_i, v_j) \in E \mbox{\ and\ } j < i\}\\
N^+(G,o,i) &=& \{v_j |(v_i, v_j) \in E \mbox{\ and\ } j > i\}
\end{eqnarray*}
The vertex order $o$ is said to be a Perfect Elimination Ordering (PEO) if
for all $i \in [n]$, $N^+(G,o,i)$ induces a complete subgraph of $G$.
The structure of $G$ can then be characterized by a tree $T(G,o) = (V_T, E_T)$
defined as follows:
\begin{eqnarray*}
V_T &=& V \\
E_T &=& \{(v_i, v_j) \in E | i < j \mbox{\ and\ } \forall j', i < j' < j, (v_i, v_{j'})\notin E \}
\end{eqnarray*}
Such a $T(G,o)$ is called the Elimination Tree associated with the graph $G$.
\end{definition}

For more information on Chordal graphs and PEO, we refer the reader to Golumbic's book~\cite{golumbic}.

\begin{definition}[Cycle Cover]\label{def:ccover}
A \textit{cycle cover} $\calC$ of $G = (V,E)$ is a set of vertex-disjoint cycles that cover the vertices of $G$. 
I.e., $\calC = \{C_1, C_2, \ldots, C_k\}$, where $V(C_i) = \{c_{i_1}, \ldots, c_{i_{r(i)}}\} \subseteq V$ such that
$(c_{i_1}, c_{i_2})$, $(c_{i_2}, c_{i_3})$, $\ldots$, $(c_{i_{r(i)-1}}, c_{i_{r(i)}})$, $(c_{i_{r(i)}}, c_{i_1}) \in E(C_i) \subseteq E$ and 
$\uplus_{i=1}^k V(C_i) = V$. The least numbered vertex $h_i \in V(C_i)$, is called the head of the cycle.
\end{definition}

\begin{definition}[\ShSACo]\label{def:shsaco}
\ShSACo\ is the class of functions from $\{0,1\}^n$ to nonnegative integers computed by
polynomial-size logarithmic-depth, semi unbounded arithmetic circuits\footnote{Note that 
such circuits have degree that is at most a polynomial in the number of input variables.}, 
using $+$ (unbounded fan-in) and $\times$ gates (fan-in $2$) and the constants $0$ and $1$.
% \Gl\ is the class of functions from $\{0,1\}^n \to \mathbb{Z}$ computed by polynomial-size
% skew circuits\footnote{A $\times$ gate in an arithmetic circuit is said to be skew if all but
% one of its children are input variables or constants. A circuit in which all $\times$ gates
% are skew is called a skew arithmetic circuit.} using bounded fan-in $+$ and $\times$ gates 
% and the constants $0$ and $1$. Alternately, they are exactly the class of functions that can 
% be expressed as a difference of accepting and rejecting computations of a Nondeterministic
% Logspace Turing machine. \Gl\ is also has the nice characterization that it is the class of 
% problems that are Logspace-reducible to the Determinant.
\end{definition}

For further background on circuit complexity, we refer the reader to~\cite{Vollmer}.

\begin{proposition}[\cite{AJMV98,Vinay91}] \label{prop:ajmv}
Any function $f: \{0,1\}^n \to R$, where $R$ is a semi-ring, computed by 
arithmetic circuits of size $s$ and degree $d$ can be computed by semi-unbounded
arithmetic circuits of size $\poly(s,d)$ and depth $O(\log{d})$. In particular,
all functions computed by polynomial sized circuits of polynomial degree are
exactly those in \ShSACo. 
\end{proposition}

\begin{fact}[Kronecker substitution~\cite{cox}]\label{fact:kron}
Let $P(x_1, x_2, \ldots, x_n)$ be a multivariate polynomial of degree $d$. We replace
every occurence of variable $x_i$ by $x^{d^i}$. This yields an unique univariate polynomial
$Q(x)$ of degree at most $d^{O(n)}$ such that $P$ can be efficiently recovered from the knowledge
of coefficients of $Q$. When the number of variables is a constant, the degree of the multivariate 
polynomial and the univariate polynomial are polynomially related.
\end{fact}

\section{From Euler Tours to Hamiltonian cycles} \label{sec:preproc}
It is possible to construct a graph $G$ such that $G$ has no Eulerian tours, but
$L(G)$ has a Hamiltonian cycle\footnote{Indeed, there is a $2$-connected graph
-- $K_4$ with one of the edges removed -- which is non-Eulerian but  
its line graph is Hamiltonian.}. Proposition \ref{prop:hnw} gives necessary and sufficient
conditions for when a line graph of a given graph is Hamiltonian. 
\begin{proposition}[\cite{HNW65}]\label{prop:hnw}
 $L(G)$ is Hamiltonian if and only if $G$ has a closed trail that contains
 at least one end point of every edge.
\end{proposition}
Given a graph $G$, we want to construct a graph $G'$ such that every closed trail in $G'$
that contains at least one end point of every edge is exactly an Eulerian tour of
$G'$. The following Lemma guarantees exactly this:
\begin{lemma}\label{lem:EHbij}
Given an undirected graph $G$, construct a graph $G' = (V', E')$ from $G$ 
as follows: Replace every edge $e = (u,v)$ of $G$ by path of length three.
Then $G$ and $G'$ have the same number of Eulerian tours and the Eulerian tours 
of $G'$ are in bijection with the Hamiltonian tours of $L(G')$.
\end{lemma}

 \begin{proof}
 Recall that $G' = (V', E')$ is obtained from $G$ as follows: Replace
 every edge $e = (u,v)$ of $G$ by path of length three, namely $(u, x_e), (x_e, y_e)$,
 $(y_e, v)$. For a graph $G$, let $\mathcal{E}_G$ and $\mathcal{H}_G$ denote the set of
 Euler Tours and Hamiltonian tours of $G$ respectively. We claim the following:

 Consider the map $h : E(G')^m \to V(L(G'))^m$ (where $m = |E(G')|$), defined by
 $h : ET \mapsto HT$. Here $ET = \langle e_1, e_2, \ldots, e_m\rangle$ is an edge sequence of $G'$
 with $e_1$ being the least edge under an arbitrary but fixed ordering of the edges of $G'$;
 $HT = \langle v_{e_1}, v_{e_2}, \ldots, v_{e_m}\rangle$ is the corresponding
 vertex sequence of $L(G')$ (where we associate the edge $e \in E(G)$ with
 vertex $v_e \in V(L(G))$). Then the proof is completed by invoking 
 Lemma~\ref{lem:eulerHam} to show that $h$ is the desired bijection 
 with its domain restricted to the set of Euler tours.
\end{proof}

 \begin{lemma}\label{lem:eulerHam}
 We have the following properties of the map $h$:
 \begin{enumerate}
 \item If $ET \neq ET'$ then $h(ET) \neq h(ET')$
 \item $h(ET)$ is defined for every Euler tour $ET$
 \item $HT = h(ET)$ is a Hamiltonian cycle for an Euler tour $ET$ 
 \item If $HT$ is a Hamiltonian cycle in $L(G')$ then there exists an Euler tour 
 $ET$ in $G'$ such that $h(ET) = HT$ 
 \end{enumerate}
 \end{lemma}
 \begin{proof}{(of Lemma~\ref{lem:eulerHam})}
 \begin{enumerate}
 \item Obvious from definitions of $h, L(G')$.
 \item Obvious from definition of $h$.
 \item If $e,e'$ are consecutive edges in $ET$, then they must share a vertex
 since $ET$ is an Euler tour and hence $v_{e}, v_{e'}$ must be adjacent
 in $L(G')$. Also since $ET$ is a permutation of all the edges of $G'$, therefore
 $h(ET)$ is a permutation of all the vertices of $L(G')$.
 \item From the way $G'$ is obtained from $G$, if $v_{e_{i-1}},v_{e_i},v_{e_{i+1}}$ are successive
 vertices on an arbitrary Hamiltonian Tour $HT$ of $L(G')$, then $e_{i-1},e_i,e_{i+1}$ cannot all be
 incident on a vertex $u \in V(G) \cap V(G')$.  For, suppose they were, then
 there exist distinct vertices $a,b,c \in V(G) \cap V(G')$, such that
 $e_{i-1},e_i,e_{i+1}$ are subdivision edges of the edges $e' = (u,a), e'' = (u,b), e''' = (u,c)$.
 But then it is easy to see that the edge $(x_{e''}, y_{e''})$ -- the middle
 subdivision edge of $e''$ -- cannot be traversed in $ET$. This is 
 since $e_i = (u,x_{e''})$, one of its only two neighbours is not used to traverse it.
 
 This implies that if $HT = v_{e_1},v_{e_2}\ldots,$ then $e_1,e_2,\dots$ is an Euler Tour
 of $G'$. Indeed, a Hamiltonian path in $L(G')$ is a permutation
of the vertices $v_e$'s of $L(G')$, thus induces a permutation of the edges
of $G'$. But a sequence $e_1,\ldots,e_m$ is an Euler tour iff for every $i$
the vertex incident on edges $e_{i-1},e_i$ and the vertex incident on edges
$e_i,e_{i+1}$ are distinct (and form the two endpoints of $e_i$) -- which follows from
the previous paragraph, completing the proof. 
 \end{enumerate}
\end{proof}

Notice that $G$ is a minor of $G'$, and the tree decomposition of $G'$ can be obtained
from that of $G$ by locally adding to each bag containing an edge $e$ of $G$, the extra
vertices and edges of the path of length three. Hence, the following is immediate:
\begin{proposition}
$G$ has bounded treewidth iff $G'$ has bounded treewidth.
\end{proposition}
\begin{proposition}[\cite{GW07}]
If $G$ is of treewidth $k$, then $L(G)$ has clique-width $f(k) = 2k + 2$.
\end{proposition}
\begin{proposition}[\cite{EJT}]\label{prop:ejt}
Given a bounded treewidth graph $G$, a balanced tree decomposition\footnote{A tree
decomposition of a graph is said to be balanced if the tree
underlying the decomposition is balanced} of $G$ is obtainable in $\Log$. %(and hence in $\SACo$).
\end{proposition}

We first need the Perfect Elimination Ordering(PEO) of the vertices of the graph. 
It is known that a graph has a PEO if and only if it is chordal. Since we can do
a chordal completion of a bounded treewidth graph (while preserving treewidth),
such an ordering of the vertices always exists. Recently Arvind et al. gave a Logspace
procedure for obtaining a PEO in $k$-trees (which are maximal treewidth-$k$ graphs).
We adapt this for graphs that are chordal completions of bounded treewidth graphs:

\begin{lemma}[Adapted from~\cite{ADKK12}]\label{lem:adkk}
Given a balanced tree decomposition of a bounded treewidth graph $G$,
a Perfect Elimination Ordering and the corresponding elimination
tree of a chordal completion of $G$, which is a balanced binary tree of depth
$O(\log{n})$, can be computed in $\Log$.
\end{lemma}

\begin{proof}
We first do a chordal completion of $G$ by adding edges to every bag in the tree
decomposition to ensure that each bag contains a simplicial vertex (it could contain more
than one, but at most $k$ since the treewidth is at most $k$). Now, we can find a
partition of the vertex set of $G$ -- $V(G) = R_0 \dotcup R_1 \dotcup \ldots \dotcup R_{l}$
as follows: First, pick one simplicial vertex from each bag in the tree decomposition
and make the layer $R_0$. If these are more than one simplicial vertex in a bag, these
are added to the sublayers of $R_0$, of which there could be at most $k$ many of them
which we call $R_{0j}$ for $j \in [k]$ (once a vertex is picked this way, it is removed
from the bag). Since the graph is chordal, this process results
in a chordal graph again, and we now do the same process iteratively, and call the 
sets of simplicial vertices so obtained, $R_1, R_2, \ldots, R_l$, each of which have
appropriate sublayers whenever there are more than one simplicial vertex in the bag (and 
we will exhaust all the vertices in the process). Note that this process can go on for at
most $l = O(\log{n})$ steps, which is the diameter of the graph (This is because we started
with a balanced binary tree decomposition of height $O(\log{n})$ and since every bag is a
clique after the chordal completion, the distance between any two nodes in this tree decomposition
is $O(\log{n})$.

Now we claim that if we order $R_{01}, \ldots, R_{lk}$ in the reverse order and within each
of these $R_i$, we order the vertices arbitrarily, we obtain a PEO of the graph. This follows
straight away from the definition of a PEO and the construction of the $R_i$'s.

Recall that an elimination tree for a graph $G = (V, E)$ and PEO $o = (v_1, v_2, \ldots, v_n)$ is
$T(G,o) = (V_T, E_T)$ and is defined as:
\begin{eqnarray*}
V_T &=& V \\
E_T &=& \{(v_i, v_j) \in E | i < j \mbox{\ and\ } \forall j', i < j' < j, (v_i, v_{j'})\notin E \}
\end{eqnarray*}

$T(G,o)$ is a tree because every vertex $v_i$, $i < n$ is adjacent to exactly one
vertex $v_j$ with $j > i$. We can now construct the elimination tree from our PEO
obtained from the $R_i$'s. Note that every vertex in the elimination tree has at 
most $k$ children which happens when a bag has $k$ vertices all of which are simplicial, 
(they are present one each in each of the sublayers $R_{ij}$, $j \in [k]$ of $R_i$).
% Since distance calculation in a tree is in Logspace, finding a PEO is in Logspace.
Hence we can construct an elimination tree of diameter at most $O(\log{n})$ in Logspace.
\end{proof}

\begin{lemma}[Adapted from~\cite{GW07}]\label{lem:gw1}
Given the tree decomposition of a graph $G$ along with a elimination tree,
the clique-width expression $X$ of $L(G)$ is obtainable in $\Log$. The parse
tree of this clique-width expression has height at most $O(\log{n})$
\end{lemma}

We show in the subsequent Lemma that the method in~\cite{GW07} 
is amenable to a Logspace implementation when provided with a PEO 
of the vertices of the graph.
\begin{lemma}[Adapted from~\cite{GW07}]\label{lem:gw}
 The NLC-width of the line graph $L(G)$ of a graph $G$ of treewidth $k$ is at most $k+2$ and
 such a NLC-width expression is obtainable in $\Log$.
\end{lemma}
Gurski and Wanke~\cite{GW07} observe that it is sufficient to look at $G$ that are $k$-trees here
because the line graph of every subgraph of $G$ then is an induced subgraph of the line graph of $G$ and the class
$NLC_k$ is closed under taking induced subgraphs for every $k \geq 1$ (See Theorem 4 in~\cite{GW07}).
Our method involves dealing with bounded treewidth graphs that are
chordal, which are a strict superclass of $k$-trees and we observe that 
the property mentioned above still holds in this case.

\begin{proof}{(of Lemma~\ref{lem:gw})}
Given an undirected graph $G = (V,E)$, let $o = (v_1, v_2, \ldots, v_n)$ 
be the PEO of the vertices of $G$. The structure of $G$ can then be
characterized by the PEO tree $T(G,o)$.
% \begin{eqnarray*}
% V_T &=& V \\
% E_T &=& \{(v_i, v_j) \in E | i < j \mbox{\ and\ } \forall j', i < j' < j, (v_i, v_{j'})\notin E \}
% \end{eqnarray*}
Let $col : V_G \to [k+1]$ be a $(k+1)$-coloring of $G$
with $col(v_i) \neq col(v_j)$ for all $(v_i, v_j) \notin E$.
Let $N^-(T(G,o), o, i) = \{v_{j_1}, v_{j_2}, \ldots, v_{j_m}\}$ (defined by the tree
$T(G, o)$) and $N^+(G, o, i) = \{v_{l_1}, \ldots, v_{l_r}\}$ (defined by the graph $G$). 
For $i = 1, \ldots, n$, an NLC-width $(k+2)$ expression is recursively defined as
follows:

\begin{enumerate}
 \item If $m = 1$, then $Y_i = X_{j_1}$. If $m > 1$, then let
 \[
  Y_i = X_{j_1} \times_I \ldots \times_I X_{j_m}
 \]
 where $I = \{(s,s) | s \in [k+1]\}$. The graph $val(Y_i)$ is the disjoint
 union of graphs $val(X_{j_1}), \ldots, val(X_{j_m})$ where vertices with the 
 same label in different graphs are connected by an edge. Note that the relation
 $I$ uses only the labels $1, \ldots, (k+1)$. The label $(k+2)$ is exclusively for
 vertices that will not be connected with other vertices in any further composition step.
 
 \item If $r > 0$, then let $Z_i$ denote a $NLC$ $(k+1)$-width expression that defines a
 complete graph with $r$ vertices labeled by $col(v_{l_1}), \ldots, col(v_{l_r})$. Here
 $r \leq k$ labels are distinct and do not include the color $col(v_i)$ of $v_i$.
 
 \item Now we define 
 \begin{equation*}
  X_i = 
  \begin{cases}
   \circ_R(Y_i \times_S Z_i) & \text{if}\ m > 0\ \text{and}\ r > 0 \\
   Z_i & \text{if}\ m = 0\ \text{and}\ r > 0 \\
   \circ_R(Y_i) & \text{if}\ m > 0\ \text{and}\ r = 0
  \end{cases}
 \end{equation*}
where,
\[
 S = \{(s,s)|s \in [k+1] - col\{(v_i)\}\} \cup \{(col(v_i),s)|s \in [k+1]\}
\]
and 
\begin{equation*}
R(s) =
\begin{cases}
 s & \text{if}\ s \neq col(v_i)\\
 (k+2) & \text{if}\ s = col(v_i)
\end{cases}
\end{equation*}
\end{enumerate}

We refer the reader to \cite{GW07} for a proof of correctness of the observation
the $NLC$-width $(k+2)$ expression $X_n$ defines the line graph of $G$. To see
that the NLC width expression $X_n$ is obtainable in $\Log$, we argue as
follows: 
\begin{enumerate}
 \item We obtain the tree decomposition of the graph $G$ in $\Log$ via \cite{EJT}.
 \item Using Lemm~\ref{lem:adkk}, we can obtain the PEO of $G$ and also construct
 the Elimination tree $T(G,o)$ in $\Log$.
 \item From $T(G,o)$ and $G$, we can obtain $m$ and $r$ and subsequently,
 each element of $N^+(G,o,i)$ and $N^-(T(G,o), o, i)$ in $\Log$.
 \item We can compute the $(k+1)$-coloring of $G$, $col : V_G \to [k+1]$ in $\Log$ via \cite{EJT}
 (Proof of Lemma 4.1).
\end{enumerate}

We build the NLC width expression for the line graph of $G$ from the
elimination tree $T(G,o)$. The NLC width expression $X_i$ is defined for
each vertex of $T(G,o)$. This however depends only on $N^-(T(G,o), o, i)$ and 
$N^+(G,o,i)$ which can be obtained in $\Log$. Along with the fact that
tree traversal via DFS is in $\Log$~\cite{CookMckenzie}, we can
obtain the NLC width expression for the line graph of $G$ in $\Log$: We can
represent the PEO tree using an expression
involving '(' and ')'. Note that such an expression can be output by a Logspace
transducer. This gives the structure of our NLC width expression, and now
we can fill in this expression using the NLC width operations. This only 
involves local computations: for example at a node $v_i$ of the tree, we compute
in Logspace $N^-(T(G,o), o, i), m$ and $N^+(G,o,i), r$ and get the appropriate
expressions based on the values of $m, r$ as given in item 3 of the NLC width
expression above. Since we build the NLC width expression over the balanced 
elimination tree of constant arity and depth $O(\log{n})$ via Lemma~\ref{lem:adkk} and 
every node in the elimination tree had atmost $k$ children, the parse tree of the NLC width
expression is also of height $O(\log{n})$.
\end{proof}

\begin{proposition}\label{prop:nlctocw}
 Given a graph $G$ of NLC-width at most $k$ by an NLC-width expression $Y$,
 we can obtain the clique-width expression $X$ of $G$, where $|X| \leq 2k + 2$ in $\Log$.
\end{proposition}

\begin{proof}
For the NLC width-$(k+2)$ expression $X_i$ defined above, there is an equivalent
clique width $(2k+2)$ expression $X_i'$. We prove by induction on $i$:
For $i=1$, there is a clique width-$(k+1)$ expression $X_1'$ because
$val(X_1)$ is just a graph on at most $k$ vertices with labels from the 
set $[k+1]$. For $i > 1$, an equivalent clique width expression $Y_i'$
for $Y_i = X_{j_1} \times_I \ldots \times_I X_{j_m}$ is obtained from the
clique width expressions for $X_{j_1}', \ldots X_{j_m}'$ and $k$ auxiliary labels.
This is because for $t = 1,\ldots, m$, the vertices of every $val(X_{j_t})$
are labeled by $k+1$ labels from $[k+2]$. Label $col(u_{j_t}) \in [k+1]$ is not 
used by the vertices of $val(X_{j_t}')$ and label $k+2$ is not involved in any
edge creation. The clique width expression $X_i'$ for $X_i = \circ_R(Y_i \times_S Z_i)$
can finally be defined by clique width expression for $Y_i$ and $k$ auxiliary labels
because $val(Z_i)$ has at most $k$ vertices.
Since all these changes are local, we can the convert the NLC width $k+2$
expression to a clique width $2k+2$ expression by replacing the corresponding
subexpressions for NLC width by the ones for clique width, to obtain the line graph of 
a bounded treewidth graph of treewidth at most $k$ in Logspace.
\end{proof}

To sum up, these are the main preprocessing steps:

\begin{enumerate}
 \item Obtain a balanced binary tree decomposition of the input 
 treewidth $k$ graph $G$ in Logspace via Proposition~\ref{prop:ejt}~\cite{EJT}.
 \item Obtain the tree decomposition of $G'$ (as required by Proposition~\ref{prop:hnw}
 and specified by Lemma~\ref{lem:EHbij}) from the tree decomposition of $G$.
 \item Perform a chordal completion of $G'$ by adding edges to every bag. 
 \item Obtain a PEO tree of $G'$ of height $O(\log{n})$, where every vertex has at most $k$
 children via Lemma~\ref{lem:adkk}.
 \item Construct a NLC width $(k+2)$ expression for $L(G')$ via Lemma~\ref{lem:gw}
 \item From the NLC width $(k+2)$ expression, construct a clique-width
 $(2k+2)$ expression for $L(G')$ via Proposition~\ref{prop:nlctocw} (The surplus edges
 added during the chordal completion are removed at this step).
\end{enumerate}

\section{The \texorpdfstring{$\#\mathsf{SAC}^1$}{} upper bound}\label{sec:algo}
Let $X$ be the clique-width $k$ expression for a labeled graph $G = (V,E, lab)$ such
that $G$ is $\val{X}$ and let $|V| = n$. Let $G$ be of clique-width $k$.
Hence by Definition~\ref{def:cw}, $G$ can be constructed from the graph with $n$
isolated labeled vertices, using at most $k$ labels. Notice that $X$ can be viewed 
as a tree (we will refer to this as the parse tree of the clique-width expression)
with the $n$ isolated labeled vertices at the leaves and every internal 
node is labeled with one of the operations
$o = \{\bullet_i, \oplus, \eta_{i,j}, \rho_{i \to j}: i,j \in [k] \wedge i \neq j\}$
To each internal vertex of the tree, we can associate a graph
(possibly disconnected) which is a subgraph of $G$, and at the root of the tree, we get
$G$ itself. The size of the tree is polynomial in $n$ and $k$. Our objective in this section
will be to count the number of Hamiltonian cycles in $G$, when provided with the 
clique-width expression $X$. We will count along the parse tree of the clique-width expression.

To this end, we call a subset of edges $E'\subseteq E$ {\em{path-cycle covers}}, if in the subgraph 
$G'=(V,E',lab)$ every vertex in $G'$ has degree at most $2$. To every such $G'$, we
associate a multiset $M$ consisting of multisets $\langle lab(v_1),lab(v_r)\rangle$
one each for every path/cycle $p=v_1,\ldots,v_r$, $r\geq 1$, in $G'$, where $v_1,v_r$
have degree at most $1$ in $G'$ if they exist ($p$ being a cycle otherwise). Let $F(X)$ be the set
of all multisets $M$ for all such subsets $E'\subseteq E$. 

Let $K$ be the set of all possible labels of the end points, in the labeled graph produced at the output
of each node in the parse tree. We refer to elements of $K$ as \emph{types}. Note that every $M$ consists of at
most $|K|$ distinct \emph{types} and $F(X)$ has at most ${(n + 1)}^{|K|}$ distinct multisets each with at 
most $n$ multisets of size $2$. 
% Let $K$ be the set of all possible labels of the
% end points, in the labeled graph produced at the output of each node in the parse tree.
% The size $K$ is $O(k^2)$. 
Here  $K = K_0 \uplus K_1 \uplus K_2$ is the set of distinct
\emph{types} where $K_2$ accounts for types of the
form $\langle i,j\rangle$ (for $i \neq j$) corresponds to paths whose end points are
$i$ and $j$; $K_0$ for the empty type $\langle\rangle = \emptyset$ corresponds to a \emph{cycle}; 
$K_1$ for types of the form $\langle i, i \rangle$ which could be either
paths whose end points are both labeled $i$, or isolated vertices with the label
$i$. Observe that, $|K_2| = \binom{k}{2}$, $K_1 = 2k$ and $K_0 = 1$, where we
distinguish between the cases of single isolated vertex of label $i$ and multiple
vertex paths with end points labeled $i$ for technical reasons, leading to the extra
factor of $2$. Our notation is consistent with~\cite{EGW01} in all cases except for
the empty type, since in~\cite{EGW01} cycles are not permitted.

Our objective is to count the number of path-cycle covers, $\#X[M]$,
corresponding to a multiset $M$ in the graph $val(X)$. In particular,
\[
\sum_{i,j \in [K]}{\#X[M_{i,j}]}
\]
 where  $M_{i,j} = \langle\langle i,j\rangle\rangle$ is a multiset
containing a single type $\langle i,j \rangle$, 
yields the number of Hamiltonian paths with end points coloured $i,j$ in $val(X)$.
We denote by $\#X$ the vector indexed by $M$ and hence has ${(n + 1)}^{K}$ entries
where $\#X[M]$ (where $M \in [0, n]^{K}$) stores the count of the number of path/cycle
covers of type specified by $M$ in the graph $val(X)$.
Let $C_o$ be a ${(n + 1)}^{K} \times {(n + 1)}^{K}$ matrix which for each pair of 
multisets $M,M'$ denotes the number of ways to form $M'$ from $M$ under an operation 
$o \in \{\eta_{i,j}, \rho_{i \to j}: i,j \in [k] \wedge i \neq j\}$.
$C_o$ is defined uniquely for the two kinds of operations $\eta,\rho$ and 
is independent of the input graph $\val{X}$. 

Then the following is an easy consequence of the definitions:
\begin{proposition}\label{prop:sharpX}
 The value of $\#X$ is given by:
\begin{enumerate}
\item if $X = \bullet_i$ then if  $M = \langle\langle i,i \rangle\rangle$ then $\#X[M] = 1$; else $\#X[M] = 0$.
\item else if $X = X_1 \oplus X_2$ then 
\[
\#X[M] = \sum_{M' \in [0,n]^{K}: M' \subseteq M}{\#X_1[M']  \#X_2[M \setminus M']}
\]
\item else if $X = \rho_{i\to j}(X_1)$ then $(C_{\rho_{i\to j}})^T \#X_1$
\item else $X = \eta_{i,j}(X_1)$ then $(C_{\eta_{i,j}})^T \#X_1$
\end{enumerate}
\end{proposition}
\begin{proof}
The first item is immediate. For the second, notice that each multiset of types $M$ in the
disjoint union of two graphs is formed by picking multisets $M',M''$ from the two graphs 
respectively and taking their multi-union. Thus the number of distinct ways to form $M$
is obtained by considering all possible decompositions of $M$ into sets $M',M''$ one from each graph. Since,
this is a decomposition $M'' = M \setminus M'$, the correctness of the second item follows.

For the third and the fourth items, notice that we have a matrix $C$ such that $C[M,M']$
is the number of ways to convert a multiset $M$ to a multiset $M'$. Thus the number of ways 
to form $M'$ is to take the product of $\#X[M]C[M,M']$ and add up the products over all $M$.
This is the stated form in matrix notation. 
\end{proof}
Proposition~\ref{prop:sharpX} enables us to prove the $\ShSACo\ $ upper bound:
\begin{lemma}\label{lem:main}
For a bounded clique-width expression $X$, for every multiset of types, $M$,
the value $\#X[M]$ of the number of path-cycle covers at any node along
the parse tree of the clique-width expression can be computed in \ShSACz\ 
where the inputs to the \ShSACz\ circuit are entries of the matrix $C_{o}$ for
$o \in \{\bullet_i, \oplus, \eta_{i,j}, \rho_{i \to j}: i,j \in [k] \wedge i \neq j\}$.
The number of path-cycle covers in the input graph can hence be counted in \ShSACo.
\end{lemma}
% \begin{lemma}\label{lem:gapL}
%   Given a clique-width expression $X$, which is the line graph $L(G)$
%   of a bounded treewidth graph $G$ with a balanced binary tree decomposition,
%   the number of Hamiltonian cycles in $L(G)$ can be computed in $\Gl$.
% \end{lemma}
\begin{proof}{(of Lemma~\ref{lem:main})}
We will use Proposition~\ref{prop:sharpX} for every node in the clique-width decomposition
$X$ to compute $\#X[M]$ for every $M \in [0,n]^K$. For this we will need the various matrices
$C_{\rho_{i\to j}}$ (constructed in Proposition~\ref{prop:Crho})
and $C_{\eta_{i,j}}$ (constructed in Lemma~\ref{lem:etaMatrixInL}). The correctness of this
circuit is clear from Proposition~\ref{prop:sharpX}.

Next, we need to argue that the number of path-cycle covers in a bounded clique-width graph
is a function in $\ShSACo$. We construct our \ShSACo\ circuit by combining the \ShSACz\ circuits
for every node in the parse tree as given by Proposition~\ref{prop:sharpX}. Notice that the 
resulting circuit is monotone (there are no subtractions) and the value at every gate is
at most a polynomial (in $n$) many bits -- this is because the number of path-cycle covers of
an $n$-vertex simple graph is at most an exponential function of $n$, which is representable by
$\poly(n)$ many bits. Thus the degree
\footnote{The degree of a circuit is defined inductively as follows: All input variables (here these
correspond to the vertices/ edges in the graph) and constants have degree $1$. The degree of a
$\times$-gate is the sum of the degree of its children. The degree of a $+$-gate is the maximum
of the degrees of its children.} of our circuit must also be polynomial in $n$. The circuit obtained thus
is of arbitrary ($\poly(n)$) depth (since the parse tree is not necessarily balanced), and hence
a naive evaluation of such a circuit is in $\Pt$. However, by Proposition~\ref{prop:ajmv} this 
circuit can be depth-reduced to yield an upper bound of 
$\ShSACo\ \subseteq \NCt \subseteq \Pt \cap \mathsf{DSPACE}(\log^2{n})$.
% \samir{This must be argued carefully}
\end{proof}

We now turn to the proof of our main Theorem~\ref{thm:main}

\begin{proof}(of Theorem~\ref{thm:main})
To count Euler tours on bounded treewidth graphs, we can count Hamiltonian cycles in
the line graph (via Lemma~\ref{lem:EHbij}). Here
we need to compute the quantity $\#X[\langle\emptyset\rangle]$ (since the empty multiset
represents a cycle, the path-cycle cover consisting of a single cycle must be a Hamiltonian
cycle itself). This follows from Lemma~\ref{lem:main}.
\end{proof}

\begin{proof}(of Theorem~\ref{thm:bonus}) Hamiltonian cycles 
can be counted in \ShSACo\ by Lemma~\ref{lem:main}. Longest Cycles (Paths) can be counted by considering 
multisets which consist of a single cycle (respectively, path) and the minimum number of isolated vertices respectively.
To see this observe that for every cycle (respectively, path) $C$ in the graph there is a multiset consisting
of a single empty type (respectively, non-empty type) and $|V(G)| - |V(C)|$ isolated vertices respectively. 

Counting cycle covers is equally simple. We just need to add up the counts for multisets
consisting only of empty types. This, is of course because an empty type represents a 
cycle.

Perfect Matchings in bipartite graphs can therefore be counted by counting the cycle 
covers in a biadjacency matrix.
\end{proof}

\subsection{Computing \texorpdfstring{$C_{\rho_{i\to j}}$}{} and  \texorpdfstring{$C_{\eta_{i,j}}$}{}}
It is easy to compute $C_{\rho_{i \to j}}$ by the following, 
\begin{proposition}\label{prop:Crho} 
$C(\rho_{i \to j})$ is a $\{0,1\}$-matrix such that the entry corresponding 
to $M_1,M_2$ is equal to $1$ iff $\rho_{i \to j}(M_1)=M_2$ (it is $0$ otherwise).
\end{proposition}
Let $W_{\vec{\alpha}}(t')$ denote the number of ways to form one path/cycle of type 
$t' \in K$, given a multiset of paths/cycles consisting of $\vec{\alpha}(t)$ 
paths/cycles for every type $t \in K$.

Next, we show how to compute $C_{\eta_{i,j}}$:
\begin{lemma}\label{lem:etaMatrixInL}
 There is a Logspace Turing machine that takes input 
$W_{\vec{\alpha}}(t')$ for every $\vec{\alpha} \in [0, n]^{|K|}, t' \in K$
and outputs the entries of the matrix $C_{\eta_{i,j}}$. 
\end{lemma}
\begin{proof}
We show that each entry can be computed in $\mathsf{DLOGTIME}$-uniform 
\TCz\ which is contained in \Log\ (see e.g. Vollmer~\cite{Vollmer}).
 Our main tool in this lemma is an application of polynomial interpolation.

Notice that the rows/columns of the matrix $C$ are indexed by multisets of 
\emph{types}. Here a \emph{type} is an element from $K$. 
Therefore any such multiset
can be described by a vector $\vec{\alpha}$ of length $|K|$. Here each
entry of the vector represents the number of paths/cycles with that type inside 
the multiset.

In the following we will consistently make use of the notation,
${\vec{z}}^{{\vec{a}}}$ to denote: $\prod_{i \in I}{z_i^{a_i}}$,
where $I$ is the index set for both $\vec{z}, \vec{a}$.

We have the following:
\begin{lemma}\label{lem:coeffPoly}
$C[M,M']$ is the coefficient of 
$\vec{x}^{\vec{c'}}\vec{y}^{\vec{c}}$ in the following polynomial
$p_{\vec{c'}, \vec{c}}(\vec{x},\vec{y})$:
\[
\prod_{t,t' \in K}{\prod_{\vec{\alpha} \in [0,n]^{|K|}}{\sum_{\vec{d}_{\vec{\alpha}}}{\left(W_{\vec{\alpha}}(t')x_{t'}y_{t'}^{\alpha(t)}\right)}}^{d_{\vec{\alpha}}(t')}}
\]
\end{lemma}
To fix the notation we reiterate (items $1,2,3,4$ were defined previously and we introduce some new notation
in items $5,6,7,8,9$):
\begin{enumerate}
\item $K$ is the set of types, where $|K| =\binom{k}{2} + 2k + 1$.
\item $t,t' \in K$  are types in the input, output multiset (respectively $M,M'$).
\item A \emph{allocation}, $\alpha(t) \in [0,n]$ is the number of path-cycle covers of type $t \in K$.
\item $\vec{\alpha} \in [0,n]^{|K|}$ is a possible allocation vector indexed by $K$ in which each entry is 
$\alpha(t)$.
\item ${d}_{\vec{\alpha}}(t') \in [0,n]$ is the number of paths of type $t'$ formed from each allocation
of type $\vec{\alpha}$.
\item $\vec{d}_{\vec{\alpha}} \in [0,n]^{|K|}$ is a vector indexed by $K$ in which each entry is one of
${d}_{\vec{\alpha}}(t')$.
\item $W_{\vec{\alpha}}(t')$ is the number of ways to form a single path/cycle of type $t'$ from an allocation
vector $\vec{\alpha}$.
\item $\vec{W}_{\vec{\alpha}}$ is the vector indexed by $K$ in which each entry is one of
$W_{\vec{\alpha}}(t')$.
\item $\vec{c}, \vec{c'} \in [0,n]^{|K|}$ are vectors indicating number of paths/cycles in $M,M'$ respectively.
\end{enumerate}
To see that Lemma~\ref{lem:etaMatrixInL} follows from Lemma~\ref{lem:coeffPoly}
we use Kronecker substitution (see Fact~\ref{fact:kron}) to convert the multivariate polynomial
$p_{\vec{c'}, \vec{c}}(\vec{x},\vec{y})$ with $2|K|$ variables 
to a univariate polynomial. Then we use Lagrange interpolation to find
the coefficient of an arbitrary term - in particular, the term corresponding
to $\vec{x}^{\vec{c'}}\vec{y}^{\vec{c}}$ in \TCz\,
(see e.g. Corollary~6.5 in~\cite{HAB}).
\end{proof}

\begin{proof}(of Lemma~\ref{lem:coeffPoly})
Consider the following expression:
\begin{equation}\label{eqn:interpMain}
\sum_{\vec{d}\in \mathbf{D}}{\prod_{t' \in K}{\prod_{\vec{\alpha} \in [0,n]^{|K|}}{{W_{\vec{\alpha}}(t')}^{d_{\vec{\alpha}}(t')}}}}
\end{equation}
where the sum is taken over $\mathbf{D} \subseteq [0,n]^{|K|}$ consisting of all $\vec{d}$'s satisfying:
\begin{eqnarray}
\label{eqn:condOne}
\forall t', \sum_{\vec{\alpha} \in [0,n]^{|K|}}{d_{\vec{\alpha}(t')}} & = & c'(t') \\
\label{eqn:condTwo}
\forall t, \sum_{\vec{\alpha} \in [0,n]^{|K|}}{\sum_{t' \in K}{{\alpha}(t){d_{\vec{\alpha}}}(t')}} & = & {c}(t)
\end{eqnarray}
\begin{claim} $C[M,M']$ equals Expression~(\ref{eqn:interpMain}).
\end{claim}
\begin{proof}(of Claim)
The Condition~\ref{eqn:condOne} above asserts that the number of paths/cycles of type $t'$ present 
in $M'$ equals the sum over all $\vec{\alpha}$ of the number of paths/cycles of type $t'$ using resources 
described by $\vec{\alpha}$; the Condition~\ref{eqn:condTwo}  is essentially a conservation
of resource equation for every type $t$ saying that all the resources present in $M$ are used one 
way or the other in $M'$. 

Let $P,P'$ be path-cycle covers represented by $M,M'$ respectively such that we 
can obtain $P'$ from $P$, i.e., $P'$ is one possible path-cycle cover that can 
be obtained by an $\eta_{i,j}$ operation on $P$. This
transformation is described by a unique $\vec{d}$. 
 Then the pair contributes precisely one to $C[M,M']$.
On the other hand $\{P,P'\}$ satisfies 
\eqref{eqn:condOne},\eqref{eqn:condTwo} so contributes 
exactly one to the summand corresponding to the unique
$\vec{d}$ in Expression \ref{eqn:interpMain}. Since the pair $P,P'$ corresponds to
a unique $\vec{d}$ and contributes exactly one, the remaining summands would evaluate
to zero. This can be explained by observing that for all $\vec{d'} \neq \vec{d}$, the
number of paths of type $t$ in $\vec{d'}$ is not equal to the corresponding number in $\vec{d}$
for atleast one $t$. Hence, they
would contribute nothing to pair $P,P'$. 
\end{proof}

To complete the proof notice that the coefficient of $\vec{x}^{\vec{c'}}\vec{y}^{\vec{c}}$
is precisely expression~\ref{eqn:interpMain} under the conditions \ref{eqn:condOne},
\ref{eqn:condTwo}. Now, we explain the reasoning behind expression~\ref{eqn:interpMain}.
We have $d_{\vec{\alpha}}(t')$ paths of type $t'$, each of which can be formed in
$W_{\vec{\alpha}}(t')$ ways. Note that each of these $d_{\vec{\alpha}}(t')$ paths are 
formed from different $\vec{\alpha}$ (though the values of each of these 
$d_{\vec{\alpha}}(t')$ vectors $\vec{\alpha}$ is the same, they are inherently different 
as they are composed of mutually exclusive vertex sets) 
and we consider each valid set of $d_{\vec{\alpha}}(t')$ vectors $\vec{\alpha}$, exactly once.
Hence, we multiply with ${W_{\vec{\alpha}}(t')}^{d_{\vec{\alpha}}(t')}$ to get the final count.
\end{proof}

\subsection{Calculation of \texorpdfstring{${W}_{\vec{\alpha}}(t')$}{}}
${W}_{\vec{\alpha}}(t')$ denotes the number of ways to form exactly 
one type $t' \in K$ in $M'$ given a multiset of types consisting of 
$\vec{\alpha}(t)$ types for every type $t \in K$ in $M$. For simplicity 
of notation, let $t = \langle i,j \rangle \in K$ be a type and let
$\beta(i) = \alpha(\langle i, i \rangle) + 
\alpha^{(=0)}(\langle i, i \rangle)$ be the total
number of multisets of type $\langle i, i \rangle$, where
$\alpha(\langle i, i \rangle)$ (respectively $\alpha^{(=0)}(\langle i, i \rangle)$) 
denote paths (respectively single nodes) labeled $\langle i, i \rangle$ in $\vec{\alpha}$.
Note that this distinction is not necessary for types where the end points 
have different labels.

\begin{lemma}\label{lem:W}
For an operation $\eta_{i_0, j_0}$ in the clique-width expression 
and for any type $t' = \langle i, j \rangle$, ${W}_{\vec{\alpha}}(t')$ is given by
\[
{W}_{\vec{\alpha}}(\langle i, j \rangle) = {\dbl {\langle i, j \rangle}\dbr}_{\vec{\alpha}}{W}_{\vec{\alpha}}
% \Delta_{\vec{\alpha}}{(\langle i, j \rangle)}
\]
{where},
\[{W}_{\vec{\alpha}} =
\binom{\alpha(\langle i_0,j_0 \rangle) + \beta(i_0) + \beta(j_0)}{\alpha(\langle i_0,j_0 \rangle)}
\alpha(\langle i_0,j_0 \rangle)!  \beta(i_0)!  \beta(j_0)!  2^{\alpha(\langle i_0,i_0 \rangle) + 
\alpha(\langle j_0,j_0 \rangle)}
 \]
and, ${\dbl {\langle i, j \rangle} \dbr}_{\vec{\alpha}}$ is given by \footnote{In this section, the notation
$\dbl S \dbr$ represents the Boolean value of the statement $S$. ${\dbl t \dbr}_{\vec{\alpha}}$ represents
a Boolean valued normalizing factor associated with the type $t$ under the allocation vector $\vec{\alpha}$.},
\footnote{We adopt a convention in which types $t'$ (other than the type $\langle i_0,j_0\rangle$)
 not explicitly included in the expressions have an allocation $\alpha_{t'}$ equalling zero.}

\begin{itemize}
\item ${\dbl {\langle a, b\rangle} \dbr}_{\vec{\alpha}} = \dbl \beta(a) = \beta(b) \dbr$
\item ${\dbl {\langle a, a\rangle}^{(=0)} \dbr}_{\vec{\alpha}} = {\dbl \alpha^{(=0)}_{a,a} = 1 \wedge \alpha(\langle a, b \rangle) = 0 \dbr}$
\item ${\dbl {\langle a, a\rangle} \dbr}_{\vec{\alpha}} = \dbl \beta(a) = \beta(b) + 1 \dbr$
% \begin{align*}
% \Delta_{\vec{\alpha}}{(\langle a,a\rangle)} &=
% [\alpha(\langle a, b \rangle) = 0 \wedge \beta(a) = \beta(b) + 1 \wedge (\alpha(\langle a, a \rangle)\geq 1\vee \alpha(\langle a, a \rangle)^{(=0)}>1)]/2 \\
% &+ [\alpha(\langle a, b \rangle) \neq 0 \wedge \beta(a) = \beta(b) + 1]
% \end{align*}
\item ${\dbl {\langle i, a\rangle} \dbr}_{\vec{\alpha}} = 
    \dbl (\alpha(\langle i, a \rangle)=1 \wedge \beta(a)=\beta(b)) \vee (\alpha(\langle i, b \rangle)=1 \wedge \beta(a)=\beta(b)+1) \dbr$
\item ${\dbl {\langle i, j\rangle} \dbr}_{\vec{\alpha}} = 
    \dbl (\alpha(\langle i, a \rangle)=1 \wedge \alpha(\langle a, j \rangle)=1 \wedge \beta(a)=\beta(b) + 1) 
    \vee (\alpha(\langle i, a \rangle)=1 \wedge \alpha(\langle b, j \rangle)=1 \wedge \beta(a)=\beta(b)) \dbr$
\item ${\dbl {\langle i, i\rangle} \dbr}_{\vec{\alpha}} = 
    \dbl (\alpha(\langle i, a \rangle)=2 \wedge \beta(a)=\beta(b) + 1) 
    \vee (\alpha(\langle i, a \rangle)=1 \wedge \alpha(\langle i, b\rangle)=1 \wedge \beta(a)=\beta(b)) \dbr$    
\item ${\dbl {\langle \emptyset \rangle} \dbr}_{\vec{\alpha}} =  \dbl \beta(a) = \beta(b) \dbr$
\end{itemize}

where, $\{a,b\} = \{i_0,j_0\}$ in some order.
\end{lemma}

\begin{proof}(of Lemma~\ref{lem:W})
Let's look at ${W}_{\vec{\alpha}}(\langle a,a \rangle)$ in detail. The ${W}_{\vec{\alpha}}(t)$
for all the other types $t$ are computed similarly. Type $\langle a,a \rangle$ can be formed 
from the alternating sequence of types $\langle a,a \rangle,\langle b,b \rangle,\langle a,a 
\rangle \ldots \langle a,a \rangle$ interleaved with some (possibly zero) $\langle a,b \rangle$ 
types. Thus, the equality $\beta(a) = \beta(b)+1$ should hold 
while $\alpha(\langle a, b \rangle)$ can be any
arbitrary non-negative integer. When $\alpha(\langle a, b \rangle)=0$, the condition
$\alpha(\langle a, a \rangle)\geq 1\vee \alpha(\langle a, a \rangle)^{(=0)}>1$ should hold to ensure that we are not 
considering the type $(\langle a,a\rangle)^{(=0)}$. 

The number of ways of interspersing $\alpha(\langle a, b \rangle)$ types among $\beta(a) + \beta(b)$ types is 
% \begin{equation*}
\[
\binom{\alpha(\langle a, b \rangle)+\beta(a)+\beta(b)}{\alpha(\langle a, b \rangle)}
\]
% \end{equation*}
We can do this for all permutations of the $\langle a,b \rangle,\langle a,a \rangle$ and
$\langle b,b \rangle$ types hence we multiply by:
%\begin{equation*}
$\alpha(\langle a, b \rangle)!  \beta(a)!  \beta(b)!.$
%\end{equation*}
Finally, we can flip the orientation of paths of types
$\langle a,a \rangle$ and $\langle b,b \rangle$ as they are equivalent
respectively to their flipped orientations. Note that single nodes cannot
be flipped. The proof is therefore completed by multiplying with:
%\begin{equation*}
$2^{\alpha(\langle a, a \rangle)+\alpha(\langle b, b \rangle)}.$
%\end{equation*}
Lastly, a boundary case occurs when $\alpha(\langle a, b \rangle)=0$ where every path can be flipped.
Here, it is easy to see that in considering every permutation of types while 
accounting for flips, we end up counting each path twice (including its reverse).
Hence, in this case we divide by $2$.  
\end{proof}

\section{Conclusion and Open Ends}\label{sec:concl}
\begin{itemize}
\item Can the \ShSACo\ bound be improved, to say, \Gl\ or Logspace?
\item How far can the Euler tour result be extended? To bounded clique-width graphs? Chordal graphs?
\item Can the determinant of bounded clique-width adjacency matrices be computed in better than \ShSACo?
(it is known to be \Log-hard even for bounded tree-width graphs from \cite{BD}).
\end{itemize}

\subparagraph*{Acknowledgements}
We would like to thank K. Narayan Kumar, Aniket Mane, M. Praveen and Prakash Saivasan for 
illuminating discussions regarding this paper. We would like to thank Eric Allender, Vikraman Arvind,
Nutan Limaye, Meena Mahajan, Pierre McKenzie, Partha Mukhopadhyay, Ramprasad Saptharishi, Srikanth 
Srinivasan and V.Vinay for reading a follow-up work that resulted from this paper and for their comments, 
from which we discovered an error in a previous version of this paper.Thanks are also due to anonymous referees for 
several comments that helped improve the content and presentation of the paper. This work is
partially funded by a grant from the Infosys Foundation.

% \newpage
\bibliographystyle{plain}
\bibliography{skeleton}

\end{document}